\documentclass[12pt]{article}
\usepackage{amssymb}
\usepackage{amsthm}
\usepackage{amsmath}
\usepackage{graphicx}
\usepackage{enumerate}
\DeclareMathOperator{\C}{C}
\newtheorem{theorem}{Theorem}[section]
\newtheorem{definition}[theorem]{Definition}
\newtheorem{lemma}[theorem]{Lemma}
\newtheorem{proposition}[theorem]{Proposition}
\newtheorem{remark}[theorem]{Remark}

\title{On Boolean posets of numerical events} 
\author{Dietmar Dorninger and Helmut L\"anger}
\date{}
\begin{document}
\footnotetext[1]{Support of the research of the second author by \"OAD, project CZ~02/2019, is gratefully acknowledged.}
\maketitle
\begin{abstract}
Let $S$ be a set of states of a physical system and $p(s)$ the probability of the occurrence of an event when the system is in state $s\in S$. A function $p\colon S\rightarrow[0,1]$ is called a numerical event or alternatively, an $S$-probability. If a set $P$ of $S$-probabilities is ordered by the order of real functions it becomes a poset which can be considered as a quantum logic. In case $P$ is a Boolean algebra this will indicate that the underlying physical system is a classical one. The goal of this paper is to study sets of $S$-probabilities which are not far from being Boolean algebras, especially by means of the addition and comparison of functions that occur in these sets. In particular, certain classes of Boolean posets of $S$-probabilities are characterized and related to each other and descriptions based on sets of states are derived.
\end{abstract}
 


\section{Introduction}

In axiomatic quantum mechanics orthomodular partially ordered sets and generalizations of them are considered as ``quantum logics'' that determine the behaviour of a physical system. In particular, if the quantum logic is a Boolean algebra then one will have reason to assume that one deals with a classical physical system. The elements of a quantum logic can also be interpreted as events, and a Boolean algebra then as the equivalent of a classical field of events as known from probability theory.

Having this in mind we first recall the notion of a {\em numerical event} (cf.\ \cite{BM91}, \cite{BM93} and \cite{MT}).

Let $S$ be a set of states of a physical system and $p(s)$ the probability of the occurrence of an event, when the system is in state $s\in S$. The function $p$ from $S$ to $[0,1]$ is called a numerical event, or alternatively more precisely an {\em $S$-probability}. If ordered by the order $\leq$ of functions and as the case may be endowed with some further properties a set of $S$-probabilities becomes a partially ordered set (poset) that can be conceived as a quantum logic. In this paper we study different kinds of such quantum logics, especially those that are not far away from being Boolean algebras. For this end we provide the following notions.

Let $P$ be a set of $S$-probabilities including the constant functions $0$ and $1$, partially ordered by the order of functions. We will call $p,q\in P$ {\em disjoint}, in symbols $p\wedge q=0$, if $x\leq p,q$ for $x\in P$ implies $x=0$. Further, $p+q$ and $p-q$ shall denote the sum and difference of $p$ and $q$, respectively, considered as real functions.

\begin{definition}\label{def1}
A set $P$ of $S$-probabilities is called {\em specific} if
\begin{enumerate}[{\rm(1)}]
\item $0,1\in P$,
\item if $p\in P$ then $p':=1-p\in P$,
\item if $p,q\in P$ and $p\wedge q=0$ then $p+q\in P$.
\end{enumerate}
\end{definition}

Condition (2) seems natural in respect to dealing with probabilities, and as for (3), this condition is motivated by classical fields of events (yet for the time being limited to considering the sum of disjoint events). Conceiving such a field of events as a Boolean ring $\mathbf R$ of subsets of some set $\Omega$, with $\boldsymbol+$ the addition in $\mathbf R$ one has $A\boldsymbol{+}B=(A\cap B^c)\cup(A^c\cap B)$ for $A,B\in\mathbf R$, where $\cup$ and $\cap$ stand for the set-theoretic join and meet (i.e.\ for union and intersection, respectively), $^c$ indicates complements in $\mathbf R$ and $\Omega$ has the role of the unity $1$ of $\mathbf R$. If $A$ and $B$ are disjoint, $A\boldsymbol{+}B=A\cup B\in\mathbf R$. Further, we observe that due to $\mathbf R$ having characteristic $2$, i.e.\ $\boldsymbol+$ will be the same as $\boldsymbol-$, $1\boldsymbol{-}A=\Omega\boldsymbol{-}A=\Omega\boldsymbol{+}A=A^c\in\mathbf R$ in coincidence with condition (2).

Two $S$-probabilities $p$ and $q$ are called {\em orthogonal}, in symbols $p\perp q$, if $p\leq q'$. From condition (3) follows that $p\wedge q=0$ for $p,q\in P$ implies $p\perp q$. For orthoposets (which specific sets of $S$-probabilities in general are not) this property is known to be {\em Boolean} (cf.\ e.g.\ \cite{T91}; in connection with orthomodular posets see i.a.\ \cite{MT} and \cite{NP}). We extend this definition to posets $(P,\leq)$ with an {\em antitone involution}, i.e.\ a mapping $'$ from $P$ to $P$ such that $p\leq q$ implies $p'\geq q'$ for $p,q\in P$ and $(p')'=p$ for $p\in P$.

\begin{definition}
A poset $P$ with an antitone involution is called {\em Boolean}, if $p\wedge q=0$ implies $p\perp q$ for $p,q\in P$.
\end{definition}

According to this definition specific sets of $S$-probabilities are Boolean posets.
 
Writing $p\vee q$ for the supremum of two elements $p,q$ of a set $P$ of $S$-probabilities and denoting their infimum by $p\wedge q$ we further point out

\begin{remark}\label{rem1}
Let $P$ be a set of $S$-probabilities satisfying {\rm(1)} and {\rm(2)} and let $p,q\in P$. Then De Morgan's laws hold in $P$ in the following sense: If $p\vee q$ exists in $P$ then $p'\wedge q'$ exists in $P$ and $(p\vee q)'=p'\wedge q'$, and if $p\wedge q$ exists in $P$ then $p'\vee q'$ exists in $P$ and $(p\wedge q)'=p'\vee q'$.
\end{remark}

\begin{proof}
Let $s\in P$. First assume $p\vee q$ to exist in $P$. Then $(p\vee q)'\leq p',q'$. If $s\leq p',q'$ then $s'\geq p,q$ which means $s'\geq p\vee q$ from which we infer $s\leq(p\vee q)'$. Hence $p'\wedge q'$ exists in $P$ and $p'\wedge q'=(p\vee q)'$. The second assertion follows by duality.
\end{proof}

Finally, we recall the definitions of two structures of numerical events which we will later relate to specific sets of numerical events.

\begin{definition}\label{def2}
A set $P$ of $S$-probabilities is called a {\em generalized field of events} {\rm(}in short {\rm GFE)} {\rm(}cf.\ {\rm\cite{D12})}, if it satisfies {\rm(1)}, {\rm(2)} and
\begin{enumerate}
\item[{\rm(4)}] if $p,q\in P$ and $p\perp q$ then $p+q\in P$.
\end{enumerate}
If a {\rm GFE} satisfies
\begin{enumerate}
\item[{\rm(5)}] if $p,q,r\in P$ and $p\perp q\perp r\perp p$ then $p+q+r\in P$,
\end{enumerate}
then it is called an {\em algebra of $S$-probabilities} {\rm(}cf.\ {\rm\cite{BM91}} and {\rm\cite{BM93})}.
\end{definition} 

Condition (4) is a special case of condition (5) -- just assume $r$ to be $0$.

The goal of this paper is to characterize various classes of specific sets of $S$-probabilities, investigate their interrelations and closeness to Boolean algebras, and indicate when they will actually be Boolean algebras. Moreover, we will consider the question whether (small) sets of $S$-probabilities will belong to a Boolean subalgebra of a specific set of $S$-probabilities and we will characterize specific sets of $S$-probabilities by states.

\section{Specific sets of varying numerical events}

\begin{definition}\label{def3}
An $S$-probability $p$ is called {\em varying}, if $p$ is neither $\leq1/2$ nor $\geq1/2$ unless $p=0$ or $p=1$.
\end{definition}

The elements of an algebra of $S$-probabilities are varying (cf.\ e.g.\ \cite{DDL}), the elements of GFEs in general are not.

As for data won by experiments: That an $S$-probability is varying often comes up independently or can be achieved by adding further experimental data directed to this purpose.

Now we will turn our attention to specific sets of $S$-probabilities that are varying. An $S$-probability is called {\em complementary} if $p\wedge p'=0$ (which by Remark~\ref{rem1} is equivalent to $p\vee p'=1$). A set $P$ of $S$-probabilities with $0$ and $1$ will be called {\em complemented} if all of its elements are complementary. Further we recall that a poset $P$ with complementation $'$ which is an antitone involution is called an {\em orthoposet}.

\begin{proposition}\label{prop1}
A specific set $P$ of varying $S$-probabilities has the following properties:
\begin{enumerate}[{\rm(i)}]
\item $P$ is complemented and hence an orthoposet,
\item if $p,q\in P$ and $p\perp q$ then $p\wedge q=0$,
\item $P$ is a {\rm GFE}.
\end{enumerate}
\end{proposition}

\begin{proof}
Let $p,q,r\in P$.
\begin{enumerate}[(i)]
\item If $r\geq p,p'$ then $r'\leq p\leq r$, from which we infer $r=1$ because of $r$ being a varying $S$-probability. Therefore $p\vee p'=1$ and hence $p\wedge p'=0$.
\item If $p\perp q$ and $r\leq p,q$ then because of $p\leq q'$ we have $r\leq q,q'$ and, since $P$ being complemented, $r=0$ showing $p\wedge q=0$.
\item If $p\perp q$ then $p\wedge q=0$ according to (ii), and by condition (3), $p+q\in P$.
\end{enumerate}
\end{proof}

\begin{remark}\label{rem2}
Let $P$ be a set of $S$-probabilities satisfying {\rm(1)} and {\rm(2)}. Then all elements of $P$ are varying if and only if $P$ is complemented.
\end{remark}

\begin{proof}
If all elements of $P$ are varying then $P$ is complemented according to the proof of Proposition~\ref{prop1} (i). Conversely, assume $P$ to be complemented. Let $p\in P$. If $p\leq1/2$ then $p\leq p'$ and hence $p=p\wedge p'=0$. Dually, if $p\geq1/2$ then $p'\leq p$ and we get $p=p\vee p'=1$. This proves that every element of $P$ is varying.
\end{proof}

An orthoposet that allows a representation by a collection $\Delta$ of subsets of a set $\Omega$ such that
\begin{itemize}
\item $\emptyset,\Omega\in\Delta$,
\item if $A\in\Delta$ then $\Omega\setminus A\in\Delta$,
\item if $A,B\in\Delta$ and $A\cap B=\emptyset$ then $A\cup B\in\Delta$,
\end{itemize}
is called a {\em concrete logic} (cf.\ \cite P).

\begin{theorem}\label{th1}
The specific sets of varying $S$-probabilities are exactly the complemented Boolean {\rm GFEs}. They all are concrete logics.
\end{theorem}

\begin{proof}
Let $P$ be a set of $S$-probabilities and $p,q\in P$. First assume $P$ to be a specific set of varying $S$-probabilities. By Proposition~\ref{prop1} (i) and (iii), $P$ is a complemented GFE. Conversely, assume $P$ to be a complemented Boolean GFE. Then, as mentioned in Remark~\ref{rem2}, the elements of $P$ are varying. Further, if $p\wedge q=0$ then $p\perp q$ and thus $p+q\in P$ according to (4). By Proposition~\ref{prop1} (i) specific sets of varying $S$-probabilities are orthoposets and hence Boolean GFEs are Boolean orthoposets. As mentioned in \cite{T91}, Boolean orthoposets are concrete logics due to a proof by Navara and Pt\'ak about Boolean orthomodular posets which does not make use of orthomodularity (cf.\ \cite{NP}).
\end{proof}

Since any specific set of varying $S$-probabilities is a concrete logic, its elements can be represented by functions which have only the values $0$ or $1$. So these $S$-probabilities must be varying from the outset. If $S$ is finite, Theorem~\ref{th1} leads to the conclusion that the specific sets of varying $S$-probabilities are Boolean algebras, since finite Boolean orthoposets are Boolean algebras. So in order to distinguish between a classical and a quantum mechanical behaviour by measurements in the form of numerical events one would need data from $S$-probabilities for a continuous set $S$ of states.

Next we turn our attention towards the connection of specific sets of varying $S$-prob\-a\-bil\-i\-ties and algebras of $S$-probabilities.

\begin{lemma}\label{lem1}   
The complemented Boolean {\rm GFEs} are exactly the algebras of $S$-probabilities that are Boolean.
\end{lemma}

\begin{proof}
According to Theorem~\ref{th1} a complemented Boolean GFE is a concrete logic, and that such a GFE is an algebra of $S$-probabilities was already shown in \cite{D12}. Conversely, every algebra of $S$-probabilities that is Boolean is also a Boolean GFE, and an arbitrary algebra of $S$-probabilities is complemented (because it is an orthoposet, first ascertained in \cite{MT}).  
\end{proof}

In fact, algebras of $S$-probabilities are orthomodular posets with a full set of states, and vice versa (cf.\ \cite{MT}). Further, an orthomodular poset is Boolean if and only if it is infimum faithful (cf.\ \cite G). To be {\em infimum faithful} means that $p\wedge q$ exists if and only if $p$ and $q$ {\em commute}, i.e.\ $p=(p\wedge q)\vee(p\wedge q')$. Since denoting an algebra of $S$-probabilities $P$ as Boolean could be mixed up with $P$ being a Boolean algebra, what in general is not the case, we rather prefer the notion infimum faithful. In the light of Theorem~\ref{th1} and Lemma~\ref{lem1} we then obtain

\begin{theorem}\label{th2}
The specific sets of varying $S$-probabilities are exactly the infimum faithful algebras of $S$-probabilities.
\end{theorem}

Returning to the motivation of the definition of specific sets of $S$-probabilities by Boolean rings, in line with Theorems~\ref{th1} and \ref{th2} we can now remark:

\begin{remark}
An infimum faithful algebra of $S$-probabilities which is a Boolean algebra can be conceived as a Boolean ring if one extends $+$ to arbitrary $S$-probabilities $p$ and $q$ by assuming within the pointwise addition of the functions $p$ and $q$ that $1+1=0$, and taking $p\cdot q:=p\wedge q$ for the ring's multiplication.
\end{remark}

\section{Further classes of specific sets of $S$-probabilities}

Let $P$ be a set of $S$-probabilities. We consider the following conditions:
\begin{enumerate}
\item[(6)] If $p,q\in P$ and $p\wedge q=0$ then $p+q=p\vee q\in P$,
\item[(7)] if $p,q,r\in P$, $p\perp q\perp r$ and $p\wedge r=0$ then $p+q+r\in P$,
\item[(8)] if $p,q,r\in P$, $p\perp q\perp r$ and $p\wedge r=0$ then $p+q+r\leq1$.
\end{enumerate}

Condition (6) can be motivated by regarding a classical field of events as a Boolean ring $\mathbf R$ for which it is the case that $A\cap B=\emptyset$ for $A,B\in R$ implies $A\boldsymbol{+}B=A\cup B$ (see Introduction). For short, we will denote specific sets of $S$-probabilities that satisfy condition (6) as {\em $\vee$-specific} ({\em join-specific}) sets of $S$-probabilities. If (1), (2) and (7) hold, $P$ is called a {\em structured set of $S$-probabilities} (cf.\ \cite{DL16}), and if (1), (2) and (8) are satisfied $P$ is known as a {\em weakly structured set of $S$-probabilities} (cf.\ \cite{DL16}).

Now we define the following classes of sets of $S$-probabilities:

$\mathcal C_1$: class of specific sets of $S$-probabilities, \\
$\mathcal C_2$: class of $\vee$-specific sets of $S$-probabilities (satisfying (6)), \\
$\mathcal C_3$: class of structured sets of $S$-probabilities (for which (7) is distinctive), \\
$\mathcal C_4$: class of weakly structured sets of $S$-probabilities (characterized by (8)).

$\mathcal C_2$ is a subclass of $\mathcal C_1$, and this is also true for $\mathcal C_3$ as one can see by setting $q=0$ within (7).

\begin{lemma}\label{lem2}
We have $\mathcal C_3\subseteq\mathcal C_2\subseteq\mathcal C_4$.
\end{lemma}

\begin{proof}
Let $P$ be a set of $S$-probabilities and $p,q,r\in P$. First assume $P\in\mathcal C_3$. As already mentioned above $P$ is a specific set of $S$-probabilities. If $p\wedge q=0$ then $p+q=p\vee q$ because for $r\geq p,q$ we have $p\perp r'\perp q$ besides $p\wedge q=0$ from which we can conclude that $p+r'+q\in P$ showing that $p+q\leq r$ and hence $p+q=p\vee q$ what explains that $P\in\mathcal C_2$ and hence $\mathcal C_3\subseteq\mathcal C_2$. Now assume $P\in\mathcal C_2$, $p\perp q\perp r$ and $p\wedge r=0$. Since $p\leq q'$ and also $r\leq q'$ we obtain that $p+r=p\vee r\leq q'$ from which we infer $p+q+r\leq1$. Therefore $P\in\mathcal C_4$ and hence $\mathcal C_2\subseteq\mathcal C_4$.
\end{proof}

\begin{lemma}\label{lem3}
We have $\mathcal C_2=\mathcal C_1\cap \mathcal C_4$.
\end{lemma}

\begin{proof}
Let $P$ be a specific set of $S$-probabilities which is also a weakly structured set of $S$-probabilities and assume $p,q,r\in P$ such that $p\wedge q=0$ and $r\geq p,q$. Then $p\perp r'\perp q$ and hence $p+r'+q\leq1$, i.e.\ $p+q\leq r$ which shows $p+q=p\vee q$. Since according to Lemma~\ref{lem2} $\mathcal C_2\subseteq\mathcal C_4$ we are done.
\end{proof}

\begin{theorem}
The class of structured sets of $S$-probabilities is a proper subclass of the class of $\vee$-specific sets of $S$-probabilities which on its part is a proper subclass of the class of weakly structured sets of $S$-probabilities unless one assumes that only specific sets of $S$-probabilities are taken into account.     
\end{theorem}

\begin{proof}
As for the inclusions to be proper, in agreement with Lemma~\ref{lem2} it suffices  to consider the following two examples: \\
First, assume $|S|=2$ and define
\[
P:=\{(0,0),(0,1/2),(1/2,1/4),(1/2,1/2),(1/2,3/4),(1,1/2),(1,1)\}.
\]
Then $P\in\mathcal C_2$, but $P\notin\mathcal C_3$ since
\[
(0,1/2)\perp(1/2,1/2)\perp(1/2,1/4)\text{ and }(0,1/2)\wedge(1/2,1/4)=(0,0),
\]
but
\[
(0,1/2)+(1/2,1/2)+(1/2,1/4)=(1,5/4)\notin P.
\]
Second example: Again we assume $|S|=2$ and this time define
\[
P:=\{(0,0),(0,1/2),(1/2,0),(1/2,1),(1,1/2),(1,1)\}.
\]
Then $P\in\mathcal C_4$, but $P\notin\mathcal C_2$ since
\[
(0,1/2)\wedge(1/2,0)=(0,0),\text{ but }(0,1/2)+(1/2,0)=(1/2,1/2)\notin P.
\]
That the $\vee$-specific sets of $S$-probabilities are exactly the elements of $\mathcal C_1\cap \mathcal C_4$ is confirmed by Lemma~\ref{lem3}.  
\end{proof}

Next we will discuss the question how far $\vee$-specific sets of $S$-probabilities are away from being Boolean algebras. A first reference to this will be the subclass $\mathcal C_3$ of $\mathcal C_2$.

\begin{theorem}\label{th3}
The members of the class $\mathcal C_3$ of structured sets of $S$-probabilities are exactly the infimum faithful algebras of $S$-probabilities.
\end{theorem}

\begin{proof}
According to Lemma~\ref{lem2} the members of $\mathcal C_3$ are GFEs with the property that $p\wedge q=0$ implies $p+q=p\vee q$, from which one can conclude (cf.\ \cite{D12}) that these posets are algebras of $S$-probabilities. Since the elements of an algebra of $S$-probabilities are varying (cf.\ e.g.\ \cite{DDL}), due to Theorem~\ref{th2} the members of $\mathcal C_3$ are infimum faithful algebras of $S$-probabilities. The converse is obvious.
\end{proof}

Though $\mathcal C_3$ is a proper subclass of $\mathcal C_2$ more incisive properties have to be taken into account to distinguish $\mathcal C_2$ from Boolean algebras: E.g., if a structured set of $S$-probabilities $P$ is finite, it is a Boolean algebra, because, as already mentioned, finite Boolean orthoposets are Boolean algebras. Further, $P$ is a Boolean algebra if it is orthocomplete (cf.\ \cite{T94}). (To be orthocomplete means that the supremum of any set of pairwise orthogonal elements of $P$ has to belong to $P$.) Moreover, if $P$ is a lattice (i.e.\ $p\vee q$ and $p\wedge q$ exist for all $p,q\in P$), then we also have a Boolean algebra (cf.\ \cite{T91}). -- That $P$ is lattice-ordered can be characterized by a simple criterion: According to Theorem~\ref{th1} $P$ is a concrete logic, and as shown in \cite{DL16}, a structured set of $S$-probabilities $P$ which is a concrete logic is a lattice if and only if for all $p,q\in P$ $\max(p,q)\in P$ (the maximum of the functions considered pointwise).

There are many papers in which (arbitrary) classes of algebras of $S$-probabilities are characterized to be Boolean algebras by specifying some structural properties -- for an overview of these papers see \cite D -- and there are numerous results on Boolean orthoposets and concrete logics which can all be applied to fathom the distance between specific sets of $S$-probabilities and Boolean algebras (cf.\ i.a.\ \cite K, \cite{NP}, \cite P, \cite{T91}, \cite{T94} and \cite{T97}).

Sometimes it is not of interest if a whole Boolean structured set of $S$-probabilities $P$ is a Boolean algebra but if a (usually small) subset of $P$ belongs to a Boolean subalgebra of $P$. If this were the case this would indicate that one locally deals with a classical physical system. To answer this question the existence of some further $S$-probabilities in $P$ will have to be asked for, but the knowledge of $P$ in detail will not be important. 

So let us assume that a subset $\{p_1,\ldots,p_n\}$ of a known or hypothetically assumed structured set of $S$-probabilities $P$ is given.
If $p_1,\ldots,p_n$ are pairwise orthogonal, then there does exist a Boolean subalgebra of $P$ wherein $p_1,\ldots,p_n$ are contained, as it is well known for every subset of mutually orthogonal elements of an orthomodular poset (cf.\ \cite{DLM}), and as proved in \cite{MT} every algebra of $S$-probabilities is orthomodular, and by Theorem~\ref{th3} also structured sets of $S$-probabilities have this property. So let us suppose that $\{p_1,\ldots,p_n\}$ is an arbitrary subset of $P$.

Having in mind that the elements of $P$ can only assume the values $0$ and $1$ (cf.\ Theorems~\ref{th1}, \ref{th2} and \ref{th3}) and defining $p\cdot q$ for $p,q\in P$ by $(p\cdot q)(s)=p(s)\cdot q(s)$ for $s\in S$ one obtains that if $p\cdot q$ exists in $P$ then $p\cdot q=p\wedge q$. It is obvious then that $p^k=p$ for $k=1,2,3,\ldots$ and that the multiplication is associative.
 
In Section~2 we have defined what it means that $p$ and $q$ commute. We will express this fact by writing $p\mathrel{\C}q$ and point out that for orthomodular posets $p\mathrel{\C}q$ is equivalent to $q\mathrel{\C}p$. Further we agree to write $\bigwedge B$ for the infimum of the elements of a finite subset $B$ of $P$. Now we can prove the following

\begin{theorem}
The set $\{p_1,\ldots,p_n\}$ is contained in a Boolean subalgebra of a structured set of $S$-probabilities $P$ if and only if $p_{i_1}\cdot\ldots\cdot p_{i_n}\in P$ for all $i_1,\ldots,i_n\in\{1,\ldots,n\}$.
\end{theorem}

\begin{proof}
Assume $n=2$. Then according to Theorem~3.4 in \cite{DLM} $\{p_1,p_2\}$ is contained in a Boolean subalgebra of $P$ if and only if $p_1\barwedge p_2(:=\min(p_1,p_2))\in P$ which in our notion means that $p_1\cdot p_2\in P$. In this theorem it is also stated that $p_1\barwedge p_2(=p_1\cdot p_2)\in P$ is equivalent to $p_1\mathrel{\C}p_2$. \\
Next we make use of Corollary~2.3 in \cite{DL14} which says: Let $A$ be a subset of an orthomodular poset $P$ with $n>1$ elements. Then $A$ is contained in an Boolean subalgebra of $P$ if and only if $(\bigwedge B)\mathrel{\C}(\bigwedge D)$ for every $k\in\{1,\ldots,n-1\}$ and every $k$-element subsets $B$ and $D$ of $A$. \\
Now we assume $P$ to be our structured set of $S$-probabilities and $A=\{p_1,\ldots,p_n\}\subseteq P$. Then $\bigwedge B$ and $\bigwedge D$ are the products $p_B$ and $p_D$ of the elements of $B$ and $D$, respectively. If $p_{i_1}\cdot\ldots\cdot p_{i_n}\in P$ for all $i_1,\ldots,i_n\in\{1,\ldots,n\}$ then $p_B\mathrel{\C}p_D$ for every subset $B$ and $D$ of $A$ with $k\leq n-1$ elements since $p_B\cdot p_D$ is an element of $P$ and, as mentioned above, $p_B\mathrel{\C}p_D$ is equivalent to $p_B\cdot p_D\in P$. Thus we can conclude that the elements of $A$ are contained in a Boolean subalgebra of $P$. -- The converse is obvious.
\end{proof}   

Besides the possibility to describe sets of $S$-probabilities by structural properties one can also try to characterize them by states, as was done by M.~J.~M\c aczy\'nski and T.~Traczyk, who characterized algebras of $S$-probabilities as the posets which have a full set of states (cf.\ \cite{MT}).
 
\section{Algebraic representations of specific sets of \\
$S$-probabilities} 

We begin by extending the commonly known notion of a state to the class of bounded posets $P$ with an antitone involution. 

\begin{definition}
A {\em specific state} on a bounded poset $\mathbf P=(P,\leq,{}',0,1)$ with an antitone involution is a mapping $s$ from $P$ to $[0,1]$ satisfying the following conditions for all $p,q\in P$:
\begin{enumerate}[{\rm(S1)}]
\item $s(0)=0$ and $s(1)=1$,
\item $s(p')=1-s(p)$,
\item if $p\leq q$ then $s(p)\leq s(q)$,
\item if $p\wedge q=0$ then there exists some $r\in P$ with $r\geq p,q$ and $s(r)=s(p)+s(q)$. 
\end{enumerate}

If for $p,q\in P$ with $p\wedge q=0$ the element $p\vee q$ exists in $P$ then a specific state on $\mathbf P$ satisfying
\begin{enumerate}
\item[{\rm(S5)}] if $p,q\in P$ and $p\wedge q=0$ then $s(p\vee q)=s(p)+s(q)$,
\end{enumerate}
is called a {\em pseudostate} on $\mathbf P$ {\rm(}cf.\ {\rm\cite{DL16})}. \\
A {\em set} $T$ {\em of specific states} on $\mathbf P$ is called {\em full} if for $p,q\in P$, $s(p)\leq s(q)$ for all $s\in T$ implies $p\leq q$, and a {\em set} $T$ {\em of specific states} on $\mathbf P$ is called {\em uniform} if for disjoint $p,q\in P$ condition {\rm(S4)} is satisfied for all $s\in T$ with the very same $r$. {\rm(}With pseudostates one can take $r=p\vee q$.{\rm)}
\end{definition}

\begin{theorem}\label{th5}
Up to isomorphism, the specific sets of $S$-probabilities are exactly the bounded posets with an antitone involution having a full and uniform set of specific states.
\end{theorem}

\begin{proof}
Let $\mathbf P=(P,\leq,{}',0,1)\in\mathcal C_1$ with $P\subseteq[0,1]^S$, $a\in S$ and $p,q\in P$. Then clearly $P$ is a bounded poset with an antitone involution. We define $s_x(r):=r(x)$ for all $x\in S$ and $r\in P$. Then we have
\begin{enumerate}[(S1)]
\item $s_a(0)=0(a)=0$ and $s_a(1)=1(a)=1$,
\item $s_a(p')=p'(a)=1-p(a)=1-s_a(p)$,
\item if $p\leq q$ then $s_a(p)=p(a)\leq q(a)=s_a(q)$,
\item if $p\wedge q=0$ then $p+q\in P$, $p+q\geq p,q$ and $s_a(p+q)=(p+q)(a)=p(a)+q(a)=s_a(p)+s_a(q)$.
\end{enumerate}
Further, if $s_x(p)\leq s_x(q)$ for all $x\in S$ then $p\leq q$. Hence $\{s_x\mid x\in S\}$ is a full and uniform set of specific states on $\mathbf P$. \\ 
Conversely, let $\mathbf P=(P,\leq,{}',0,1)$ be a bounded poset with an antitone involution which has a full and uniform set $S$ of specific states and let $p,q,r\in P$ and $s\in S$. We define $(f(u))(t):=t(u)$ for all $u\in P$ and $t\in S$. Then the following assertions are equivalent: $f(p)\leq f(q)$, $(f(p))(t)\leq(f(q))(t)$ for all $t\in S$, $t(p)\leq t(q)$ for all $t\in S$, $p\leq q$. Therefore $f(p)=f(q)$ if and only if $p=q$. Next we will prove $f(P)\in\mathcal C_1$:
\begin{enumerate}[(1)]
\item $(f(0))(s)=s(0)=0$ and $(f(1))(s)=s(1)=1$; thus $0=f(0)\in f(P)$ and $1=f(1)\in f(P)$.
\item $(f(p'))(s)=s(p')=1-s(p)=1-(f(p))(s)=(f(p))'(s)$; therefore $(f(p))'=f(p')\in f(P)$.
\item Assume $f(p)\wedge f(q)=f(0)$. If $r\leq p,q$ then $f(r)\leq f(p),f(q)$, from which we infer $f(r)=f(0)$, i.e.\ $r=0$, showing $p\wedge q=0$. Accordingly, there exists some $u\in P$ (which is independent of $s$) with $s(u)=s(p)+s(q)$. Now
\[
(f(p)+f(q))(s)=(f(p))(s)+(f(q))(s)=s(p)+s(q)=s(u)=(f(u))(s),
\]
i.e.\ $f(p)+f(q)=f(u)\in f(P)$.
\end{enumerate}
From this we can conclude that $f(\mathbf P)=(f(P),\leq,{}',0,1)\in\mathcal C_1$ and that $f$ is an isomorphism from $\mathbf P$ onto $f(\mathbf P)$. Hence $\mathbf P$ is isomorphic to a member of $\mathcal C_1$.
\end{proof}

As shown in \cite{DL16} up to isomorphism the weakly structured sets of $S$-probabilities are exactly the bounded posets with an antitone involution in which the join of two disjoint elements exists and which have a full set of pseudostates, which in the light of Lemma~\ref{lem3} then reads

\begin{theorem}\label{th4}
Up to isomorphism, the $\vee$-specific sets of $S$-probabilities are exactly the bounded posets with an antitone involution in which the sum of two disjoint elements equals their join and which have a full set of pseudostates.
\end {theorem}

Theorems~\ref{th5} and \ref{th4} are analogues to the theorem mentioned above that up to isomorphism the algebras of $S$-probabilities are exactly the orthomodular posets having a full set of states.

Authors' addresses:

Dietmar Dorninger \\
TU Wien \\
Faculty of Mathematics and Geoinformation \\
Institute of Discrete Mathematics and Geometry \\
Wiedner Hauptstra\ss e 8-10 \\
1040 Vienna \\
Austria \\
dietmar.dorninger@tuwien.ac.at

Helmut L\"anger \\
TU Wien \\
Faculty of Mathematics and Geoinformation \\
Institute of Discrete Mathematics and Geometry \\
Wiedner Hauptstra\ss e 8-10 \\
1040 Vienna \\
Austria, and \\
Palack\'y University Olomouc \\
Faculty of Science \\
Department of Algebra and Geometry \\
17.\ listopadu 12 \\
771 46 Olomouc \\
Czech Republic \\
helmut.laenger@tuwien.ac.at
\end{document}